\documentclass[a4paper,UKenglish]{article}
\usepackage{tikz}
\usepackage[algoruled,linesnumbered,vlined]{algorithm2e}
\usepackage{amsthm}
\usepackage{amsmath}
\usepackage{amsfonts}
\usepackage{url}
\usepackage{hyperref}
\usepackage{geometry}
\usepackage{microtype}
\usepackage{graphicx}

\newtheorem{lemma}{Lemma}
\newtheorem{corollary}{Corollary}
\newtheorem{theorem}{Theorem}
\newtheorem{definition}{Definition}

\title{Cadences in Grammar-Compressed Strings}

\author{Julian Pape-Lange\thanks{Technische Universit\"at Chemnitz, Stra\ss e der Nationen 62, 09111 Chemnitz, Germany. Email: julian.pape-lange@informatik.tu-chemnitz.de}}

\date{}

\begin{document}

\maketitle

\begin{abstract}
Cadences are structurally maximal arithmetic progressions of indices corresponding to equal characters in an underlying string.

This paper provides a polynomial time detection algorithm for $3$-cadences in grammar-compressed binary strings. This algorithm also translates to a linear time detection algorithm for $3$-cadences in uncompressed binary strings.

Furthermore, this paper proves that several variants of the cadence detection problem are $\mathcal{NP}$-complete for grammar-compressed strings. As a consequence, the equidistant subsequence matching problem with patterns of length three is $\mathcal{NP}$-complete for grammar-compressed ternary strings.
\end{abstract}

\section{Introduction}

A sub-cadence in a string is an arithmetic progression of indices corresponding to equal characters. This concept in the context of finite sequences is quite old and dates back to van der Waerden. He showed in the year 1927 in \cite{vanDerWaerden} that for each $k$ and each alphabet size $|\Sigma|$, there is a natural number $m=m(k,|\Sigma|)$, such that each sequence of characters in $\Sigma$ with length greater than or equal to $m$ has a sub-cadence consisting of $k$ indices.

The term cadence in the context of strings was first used by Gardelle in \cite{Cadences_Gardelle} in the year 1964.

In this paper, I use the notation of Amir et al.\ in \cite{Cadences_Amir} and say that a cadence is a sub-cadence which is structurally maximal in the sense that the extension of the arithmetic progression to the left or to the right would not result in a valid index of the string.

For example, in the string $S=10101$, the three indices $(1,3,5)$ form a cadence, since the indices $-1$ and $7$ are both outside of the string. On the other hand, in the string $S=01110$, the three indices $(2,3,4)$ do not form a cadence, since, for example, the index $1$ is inside the string.

Cadences recently gained some traction. At the beginning of this year it was proven by Funakoshi and Pape-Lange in \cite{FPL_3-Cadences} that the number of $3$-cadences can be counted in $\mathcal{O}(n (\log n)^2)$ time using fast Fourier transform if the underlying alphabet has a constant size. 

More recently, Funakoshi et al.\ presented the more general problem of equidistant subsequence matching in \cite{Funakoshi_ESM} which extends the cadences to arbitrary arithmetic factors, and showed that techniques for cadence-detection can be adopted to solve equidistant subsequence matching with similar time complexity.

Strings can be compressed by straight-line programs, which are context-free grammars whose languages contain exactly one string each. Since this grammar-based compression is able to compress some strings to logarithmic size, we are interested which polynomial time problems on uncompressed strings can also be solved in polynomial time with respect to the compressed size of the string. For example, grammar-based compression allows for fast algorithms as the fully compressed pattern matching by Je\.z presented in \cite{Jez}. Also, the size of the smallest grammar is comparable to other strong string compression algorithms as LZ77 (as proven simultaneously by Rytter in \cite{Rytter_smallest_grammar} and by Charikar et al.\ in \cite{Charikar_smallest_grammar}) and thereby as the run-length encoded Burrows-Wheeler transform (as recently proven independently by Kempa and Kociumaka in \cite{Kempa_BWT} and by Pape-Lange in \cite{Pape-Lange_BWT}).

In this paper, I will prove that it can be decided in polynomial time whether a binary grammar-compressed string contains a $3$-cadence. Furthermore, this polynomial time $3$-cadence detection algorithm for binary grammar-compressed strings also translates to a linear time $3$-cadence detection algorithm for uncompressed binary strings.

In order to obtain these algorithms, this paper introduces two new special cases of the sub-cadence, the $L$-$R$-cadence, which starts and ends in given intervals, and the even/odd $3$-sub-cadence, which start at even/odd indices.

I will also prove that for grammar-compressed strings, the cadence detection problem becomes $\mathcal{NP}$-complete for longer cadences or $3$-cadences over a ternary alphabet.

\section{Preliminaries}

A string $S$ of length $n$ is the concatenation $S=S[1]S[2]S[3]\dots S[n]$ of characters from an alphabet $\Sigma$. Strings naturally split into runs of equal characters. For example, the string $00010101100$ splits into $000\cdot1\cdot0\cdot1\cdot0\cdot11\cdot00$. In this paper, these runs of equal characters are just called runs for the sake of simplicity.

For the sub-cadences and cadences, this paper uses the definitions of Amir et al.\ in \cite{Cadences_Amir}. These definitions are slightly different from the definition by Gardelle in \cite{Cadences_Gardelle} and by Lothaire in \cite{Cadences_Lothaire}. Funakoshi and Pape-Lange present a comparison of these definitions in \cite{FPL_3-Cadences}.

\begin{definition}
    A \emph{$k$-sub-cadence} is an arithmetic progression of indices given by the triple $(i,d,k)$ of integers such that $d>0$ and \[S[i]=S[i+d]=\dots=S[i+(k-1)d]\] hold.
\end{definition}

As a special case, cadences additionally have to be \emph{structurally maximal} in the sense that every extension of the underlying arithmetic progression is not contained in the integer interval $\{1,2,3,\dots, n\}$ anymore. More formally:

\begin{definition}
    A \emph{$k$-cadence} is a $k$-sub-cadence $(i,d,k)$ such that the inequalities $i-d\leq 0$ and $n<i+kd$ hold.
\end{definition}

In this paper, we will also consider a new special case of the sub-cadence, in which the first element and the last element of the sub-cadence are contained in given intervals:

\begin{definition}
    For two disjoint intervals $L$ and $R$, a \emph{$L$-$R$-$k$-cadence} is a $k$-sub-cadence which starts in the interval $L$ and ends in the interval $R$. I.e.\ $i\in L$ and $i+(k-1)d\in R$ hold.
\end{definition}

For the compressed problems, we consider straight-line grammars in Chomsky normal form. I.e.\ a string is given by a grammar $G=(V, \Sigma, S, rhs)$ such that the set $V=\{v_1, v_2, \dots, v_{|V|}\}$ of nonterminals is ordered so that for every nonterminal $v_i\in V$ the right-hand side is either a character from the alphabet or of the form $v_j v_k$ with $j,k< i$. We also assume that the start symbol $S$ is given by $v_{|V|}$.

Since each string which has at least the size of van der Waerden's bound $m=m(k,|\Sigma|)$ has a $k$-sub-cadence, we can detect a $k$-sub-cadences by restricting the string to the first $m$ characters. Since this constant $m$ is only dependent on $k$ and $\Sigma$ but not dependent on the length of the original string, the detection algorithm for $k$-sub-cadences only uses constant time on uncompressed strings.

%
%

\section{NP-Complete Cadence Problems}

In this section, I will prove the following theorem:

\begin{theorem} \label{thm:Cad-NP}
    The decision problem of $k$-cadence detection on grammar-compressed strings is $\mathcal{NP}$-complete if at least one of the following conditions holds:
    \begin{itemize}
        \item $k\geq3$ and $|\Sigma|\geq 2$ and we only consider $k$-cadences with a given character,
        \item $k\geq3$ and $|\Sigma|\geq 3$ or
        \item $k\geq4$ and $|\Sigma|\geq 2$.
    \end{itemize}
\end{theorem}

Since we can test for a given candidate $(i,d,k)$ of a $k$-cadence in polynomial time, whether $(i,d,k)$ forms indeed a $k$-cadence, all three problems mentioned above belong to $\mathcal{NP}$.

To show the $\mathcal{NP}$-hardness, I will reduce the following problem, which Lohrey proves in Theorem 3.13 of \cite{LohreyBook} to be $\mathcal{NP}$-complete, to the problems above: \\
\begin{tabular}{ll}
    \textbf{input:} & Two strings $P$ and $P'$ over the alphabet $\{0,1\}$ given by grammar-compression. \\
    \textbf{output:} & Is there an index $l$ with $P[l]=P'[l]=1$?
\end{tabular} 

Let $P$ and $P'$ be strings over the alphabet $\{0,1\}$ given by grammar-compression. Without loss of generality $|P'|\leq |P|$ holds. Since it is more convenient if both strings have the same length, we pad the shorter string $P'$ with zeros. Also, for the cadences, it will be helpful, if one of the strings is reversed. We therefore define $P''=(P' 0^{|P|-|P'|})_{\operatorname{rev}}$.

In this setting, for every index $l$, the equation $P[l]=P'[l]=1$ holds if and only if the equation $P[l]=P''[|P|+1-l]=1$ holds as well.

Consider the string
\[S = 
\Big(0^{(k-1)|P|} \cdot P \cdot 0 \cdot 0^{k|P|}\Big) 
\Big(0^{k|P|} \cdot 1 \cdot 0^{k|P|}\Big)
\Big(0^{k|P|} \cdot 0 \cdot P'' \cdot 0^{(k-1)|P|}\Big) 
\Big(1^{2k|P|+1}\Big)^{k-3}
\textup{.}\]

The grammar of $S$ can be built by the grammars of $P$ and $P'$ and $\mathcal{O}\left(\log k^2|P|\right)$ additional nonterminals. Since the grammar-compression of a string $P$ needs at least $\Omega (\log |P|)$ nonterminals, the compressed size of $S$ is, for fixed $k$, polynomial in the compressed size of the inputs.

If there is an index $l$ with $P[l]=P''[|P|+1-l]=1$, we can construct a corresponding $k$-cadence in $S$ with character $1$:

The corresponding indices $(k-1)|P|+l$ and $2(2k|P|+1)+k|P|+1+(|P|+1-l)$ to $P[l]$ and $P''[|P|+1-l]$ in $S$ as well as the $1$ in the second bracket at index $(2k|P|+1)+k|P|+1$ form an arithmetic progression starting at $i=(k-1)|P|+l$ with distance $d=2k|P|+1+(|P|+1-l)$.

The index $l$ is bounded by $1\leq l \leq |P|$, and each bracket has length $2k|P|+1$. Therefore, $n=k(2k|P|+1)$ holds and the inequalities 
\[i+kd = ((k-1)|P|+l) + k(2k|P|+1+(|P|+1-l)) > k(2k|P|+1) = n\]
and
\begin{align*}
i+(k-1)d &= ((k-1)|P|+l) + (k-1)(2k|P|+1+(|P|+1-l)) \\
&\leq (k|P|) + (k-1)(2k|P|+1) + (k-1)|P| < k(2k|P|+1) = n
\end{align*}
as well as $i-d\leq 0$ and $i>0$ hold as well.

Furthermore, for $0\leq j< k$ the index $i+jd$ lies in the $(j+1)$-th bracket. Therefore, $S[i+jd]=1$ holds for each $0\leq j< k$. This implies that $(i,d,k)$ is a $k$-cadence with character~$1$.

If, on the other hand, the triple $(i,d,k)$ defines a $k$-cadence with character $1$ in $S$, we can find a corresponding index $l$ with $P[l]=P''[|P|+1-l]=1$:

The inequalities $i-d\leq 0 < i$ and $i+(k-1)d\leq n < i+kd$ of the cadence imply 
\[\frac{j}{k}n < \frac{k-j}{k}i+\frac{j}{k}(i+kd) = i+jd = \frac{k-j-1}{k}(i-d)+\frac{j+1}{k}(i+(k-1)d) \leq \frac{j+1}{k}n\textup{.}\]

Since the brackets in the definition of $S$ divide the string in $k$ substrings with equal length, the $(j+1)$-th element of any cadence lies in the $(j+1)$-th of the $k$ brackets. Therefore, each cadence with character $1$ contains the single $1$ at index $(2k|P|+1)+(k|P|)+1$ in the second bracket. Furthermore, the first element of the arithmetic progression has to be a $1$ in $P$ in the first bracket and the third element of the arithmetic progression has to be a $1$ in $P''$ in the third bracket. 

By construction, the two indices of these characters have the same distance to the index $(2k|P|+1)+(k|P|)+1$, and the two strings $P$ and $P''$ have the same distance to the index $(2k|P|+1)+(k|P|)+1$ as well. Therefore, the first element of the $k$-cadence and the third element of the $k$-cadence define an index $l$ with $P[l]=P''[|P|+1-l]=1$.

Therefore, the string $S$ has a $k$-cadence with character $1$ if and only if there is an index $l$ such that $P[l]=P'[l]=1$ holds.

If $k>3$ holds, there is at least one bracket in $S$ containing only the character $1$. In this case, this bracket forces every $k$-cadence to be a $k$-cadence with character $1$. Therefore, in this case, the requirement that the underlying character has to be $1$ can be dropped.

For $3$-cadences on a ternary alphabet we consider the string 
\[S = 
\Big(0^{(k-1)|P|} \cdot P \cdot 0 \cdot 0^{k|P|}\Big) 
\Big(2^{k|P|} \cdot 1 \cdot 2^{k|P|}\Big)
\Big(0^{k|P|} \cdot 0 \cdot P'' \cdot 0^{(k-1)|P|}\Big) 
\textup{.}\]

Since the first and the last bracket do not contain the character $2$, there are no $k$-cadences with character $2$. Since the second bracket does not contain the character $0$, there are no $k$-cadences with character $0$ either. Therefore, all $k$-cadences use the character $1$, and there is a $3$-cadence in $S$ if and only if there is an index $l$ with $P[l]=P'[l]=1$.

This concludes the proof of Theorem \ref{thm:Cad-NP}.

\section{L-R-Cadences} \label{sec:LR}

In this section, I will show that the problems discussed in the last section are also $\mathcal{NP}$-complete for $L$-$R$-cadences instead of cadences. Even if $k=3$ and $|\Sigma|=2$ hold, the compressed detection problem of $L$-$R$-$k$-cadences is $\mathcal{NP}$-complete. However, in this special case, there is a polynomial time detection algorithm if the length of $L$ is similar to the length of $R$. The underlying idea for this algorithm also leads to a linear time algorithm for the detection of $L$-$R$-$3$-cadences in uncompressed binary strings.

In uncompressed strings, the first proposed detection algorithm for $3$-cadences by Amir et al.\ in \cite{Cadences_Amir} was actually a detection algorithm for $L$-$R$-$3$-cadences with $L=\left\{1,2,\dots,\left\lfloor\frac{1}{3}n\right\rfloor\right\}$ and $R=\left\{\left\lfloor\frac{2}{3}n\right\rfloor+1, \left\lfloor\frac{2}{3}n\right\rfloor+2,\dots, n\right\}$. Furthermore, the algorithm of Funakoshi and Pape-Lange in \cite{FPL_3-Cadences} count the number of $3$-cadences in $\mathcal{O}\left(n(\log n)^2\right)$ time, by counting $L$-$R$-$3$-cadences in $\mathcal{O}\left((|L|+|R|)(\log (|L|+|R|))\right)$ time. It therefore seems reasonable to understand the $L$-$R$-cadences to be a simplification of cadences.

However, for all cadence problems discussed in the last section, the corresponding $L$-$R$-cadence problem is $\mathcal{NP}$-complete too:

\begin{lemma} \label{lem:LRCad-NP}
    The decision problem of $L$-$R$-$k$-cadence detection is $\mathcal{NP}$-complete on grammar-compressed strings if at least one of the following conditions holds:
    \begin{itemize}
        \item $k\geq3$ and $|\Sigma|\geq 2$ and we only consider $L$-$R$-$k$-cadences with a given character,
        \item $k\geq3$ and $|\Sigma|\geq 3$ or
        \item $k\geq4$ and $|\Sigma|\geq 2$.
    \end{itemize}
\end{lemma}

The proofs are essentially equal to the corresponding proofs in the last section, since for $L=\left\{1,2,\dots,\frac{1}{k}n\right\}$ and $R=\left\{\frac{k-1}{k}n+1,\frac{k-1}{k}n+2,\dots,n\right\}$,  all $k$-cadences in the discussed string $S$ are $L$-$R$-$k$-cadences and vice versa.

Next, I will show that in the case $k=3$ and $|\Sigma|=2$, even if we do not require a given character, the decision problem of $L$-$R$-$k$-cadence detection is $\mathcal{NP}$-complete on grammar-compressed strings:

Since we can test for every triple $(i,d,k)$, whether this triple forms an $L$-$R$-$k$-cadence, this problem belongs to $\mathcal{NP}$.

To show the $\mathcal{NP}$-hardness, we will, like in the last section, reduce the following $\mathcal{NP}$-complete problem to the decision problem of $L$-$R$-$3$-cadence detection in grammar-compressed strings: \\
\begin{tabular}{ll}
    \textbf{input:} & Two strings $P$ and $P'$ over the alphabet $\{0,1\}$ given by grammar-compression. \\
    \textbf{output:} & Is there an index $l$ with $P[l]=P'[l]=1$?
\end{tabular} 

Let $P$ and $P'$ be strings over the alphabet $\{0,1\}$ given by grammar-compression. Without loss of generality $|P'|\leq |P|$ holds. Since it is more convenient if both strings have the same length, we pad the shorter string $P'$ with zeros. Also, for the $L$-$R$-$k$-cadences, it will be helpful, if in one of the strings, each character is duplicated. For example, for $P'=011$, we define $P''=001111$. This can be done by introducing two additional nonterminals.

Define $S=1(0^{|P|})(P)(P'')$, $L=\{1\}$ and $R=\{1+2|P|+1, 1+2|P|+2, \dots 1+2|P|+2|P|\}$. In this setting $S[L]=1$ and $S[R]=P''$ holds. Furthermore, for each index $1\leq l \leq |P|$, the equations $P[l]=S[1+(|P|+l)]$ and $P'[l]=P''[2l]=S[1+2|P|+2l]=S[1+2(|P|+l)]$ hold.

Therefore, for each index $l$, the equation $P[l]=1=P'[l]$ holds if and only if the equation $S[1]=S[1+(|P|+l)]=S[1+2(|P|+l)]$ holds. This equation, however, defines an $L$-$R$-$3$-cadence.

This proves, that $S$ has an $L$-$R$-$3$-cadence if and only if there is an index $l$ such that $P[l]=P'[l]=1$ holds.

Together with the previous lemma, this implies:

\begin{theorem} \label{thm:LRCad-NP}
    For $k\geq3$ and $|\Sigma|\geq 2$, the decision problem of $L$-$R$-$k$-cadence detection is $\mathcal{NP}$-complete on grammar-compressed strings.
\end{theorem}

Since the equidistant subsequence matching problem is closely related to sub-cadences, we can similarly show that equidistant subsequence matching with patterns of length $3$ on ternary strings is $\mathcal{NP}$-complete on grammar-compressed strings. 

Consider the pattern $P=212$, and a string $S$ with $S[L], S[R]\in\{0,2\}^*$ and all other characters are either $0$ or $1$. Define $S'$ by
\[
S'[i] = \begin{cases}
0 & \textup{if } S[i]=0\\
1 & \textup{if } S[i]\neq 0
\end{cases}
\]
the string in which all ``$2$''s in $S$ are replaced by a ``$1$''.
In this setting, the equidistant occurrences of $P$ in $S$ are exactly the $L$-$R$-$3$-cadences with character $1$ in $S'$.

All reductions above used that we could force all cadences to use a fixed character of the string. However, surprisingly, if $L$ and $R$ have similar length, we can detect in polynomial time, whether a compressed binary string has an $L$-$R$-$3$-cadence. Furthermore, with the same idea we can detect in linear time, whether an uncompressed binary string has an $L$-$R$-$3$-cadence.

The remainder of this section proves the following theorem:

\begin{theorem} \label{thm:LRCad-P}
The decision problem of $L$-$R$-$3$-cadence detection in binary grammar-compressed strings can be solved in polynomial time with respect to the compressed size of the string and the additional variable $\max\left(\frac{|L|}{|R|},\frac{|R|}{|L|}\right)$.

The decision problem of $L$-$R$-$3$-cadence detection in binary uncompressed strings can be solved linear time with respect to $|L|+|R|$.
\end{theorem}

Since the first index and the third index of each $3$-sub-cadence have the same parity, it is useful to divide the $L$-$R$-$3$-cadences according to this parity:

\begin{definition}
    For two disjoint intervals $L$ and $R$, an \emph{even $L$-$R$-$3$-cadence} is a $3$-sub-cadence which starts at an even index in $L$ and ends at an even index in $R$.
    
    Similarly, an \emph{odd $L$-$R$-$3$-cadence} is a $3$-sub-cadence which starts at an odd index in $L$ and ends at an odd index in $R$.
    
    For each set $M$,  we define $M_{\operatorname{even}}:=M\cap 2\mathbb{Z}$ and $M_{\operatorname{odd}}:=M\cap (2\mathbb{Z}+1)$ and for each $M=\{a_1, a_2, \dots, a_l\}\subset \mathbb{Z}$ with $1\leq a_1<a_2<a_3<\dots<a_l\leq n$, we define the string $S[M]=S[a_1]S[a_2]\dots S[a_l]$ as the subsequence of characters with indices given by $M$.
\end{definition}

The key insight for the detection algorithm for $L$-$R$-$3$-cadences is that if the string does not contain $L$-$R$-$3$-cadences, either $S[L_{\operatorname{even}}]$ or $S[R_{\operatorname{even}}]$ is very structured. The following lemma implies that if $S[L_{\operatorname{even}}]$ has the substring $01$ and $S[R_{\operatorname{even}}]$ has the substring $10$ or vice versa, then $S$ has an $L$-$R$-$3$-cadence:

\begin{lemma} \label{lem:LR2}
    Let $S$ be a binary string and $L$ and $R$ be two intervals.
    
    If there are indices $i$ and $j$ with 
    \begin{itemize}
        \item $S[i]=S[j]\neq S[i+2]=S[j-2]$,
        \item $i, i+2\in L$, 
        \item $j, j-2\in R$ and 
        \item $i\equiv j \pmod 2$, 
    \end{itemize}
    then $S$ has an $L$-$R$-$3$-cadence.
\end{lemma}
\begin{proof}
    Since $i\equiv j \pmod 2$ holds, the number $\frac{i+j}{2}$ is an integer. Furthermore, since $S$ is binary and $S[i]=S[j]\neq S[i+2]=S[j-2]$ holds, we either have $S[i]=S[\frac{i+j}{2}]=S[j]$ or $S[i+2]=S[\frac{i+j}{2}]=S[j-2]$. Therefore, there is at least one $L$-$R$-$3$-cadence.
\end{proof}

This implies that if $S$ does not contain $L$-$R$-$3$-cadences, then there are only few possibilities for the subsequences $S[L_{\operatorname{even}}]$ and $S[R_{\operatorname{even}}]$:

\begin{corollary} \label{cor:LReasiness}
    Let $S$ be a binary string and $L$ and $R$ be two intervals such that $S$ has no $L$-$R$-$3$-cadences.
    
    Then,
    \begin{itemize}
        \item if $S[L_{\operatorname{even}}]$ is of the form $0^i 1^{i'}$ with $i, i'>0$, then $S[R_{\operatorname{even}}]$ is of the form $0^j 1^{j'}$ where $j$ and $j'$ may be equal to $0$,
        \item if $S[L_{\operatorname{even}}]$ is of the form $1^i 0^{i'}$ with $i, {i'}>0$, then $S[R_{\operatorname{even}}]$ is of the form $1^j 0^{j'}$ where $j$ and ${j'}$ may be equal to $0$, and
        \item if $S[L_{\operatorname{even}}]$ contains the substrings $01$ and $10$, then $S[R_{\operatorname{even}}]$ is of the form $0^{j}$ or $1^{j}$.
    \end{itemize}
\end{corollary}

We can check in linear time in uncompressed strings and in polynomial time in grammar-compressed strings whether $S[L_{\operatorname{even}}]$ and $S[R_{\operatorname{even}}]$ are of the form $0^i 1^{j}$ or $1^i 0^{j}$. If both $S[L_{\operatorname{even}}]$ and $S[R_{\operatorname{even}}]$ are of the form $0^i 1^{j}$, we can divide $L_{\operatorname{even}}$ and $R_{\operatorname{even}}$ into $L'_{\operatorname{even}}$, $L''_{\operatorname{even}}$, $R'_{\operatorname{even}}$ and $R''_{\operatorname{even}}$ such that $S[L'_{\operatorname{even}}]=0^{i}$, $S[L''_{\operatorname{even}}]=1^{i'}$, $S[R'_{\operatorname{even}}]=0^{j}$ and $S[R''_{\operatorname{even}}]=1^{j'}$. 

Since there are, by construction, no even $L'$-$R''$-$3$-cadences and no even $L''$-$R'$-$3$-cadences, we only have to detect $L'$-$R'$-$3$-cadences and $L''$-$R''$-$3$-cadences. This can be done in linear time in uncompressed strings and in polynomial time in grammar-compressed strings using the following lemma, which holds by definition of the even $L$-$R$-$3$-cadence:

\begin{lemma} \label{lem:LR01}
     Let $S$ be a binary string and $L$ and $R$ be two intervals such that $S[L_{\operatorname{even}}]=0^i$ and $S[R_{\operatorname{even}}]=0^j$ hold for some integers $i$, $j$. Let further $l_{\min} = \min(L_{\operatorname{even}})$, $l_{\max} = \max(L_{\operatorname{even}})$, $r_{\min} = \min(R_{\operatorname{even}})$ and $r_{\max} = \max(R_{\operatorname{even}})$.
     
     Then, $S$ has an even $L$-$R$-$3$-cadence if and only if 
     \[S\left[\left\{\frac{l_{\min}+r_{\min}}{2}, \frac{l_{\min}+r_{\min}}{2}+1, \dots, \frac{l_{\max}+r_{\max}}{2}\right\}\right]\neq 1^{\left(\frac{l_{\max}+r_{\max}}{2}-\frac{l_{\min}+r_{\min}}{2}+1\right)}\]
     holds.
\end{lemma}

\begin{figure}
    \centering
    \begin{tikzpicture}[scale=1]
    \draw [ultra thick] (0,4) -- (0,0);
    \draw [ultra thick] (0,4) -- (1,0);
    \draw [ultra thick] (0,4) -- (2,0);
    
    \draw (-5,4) node [fill=white,anchor=west] {$S[\{1,2,\dots,16\}_{\operatorname{even}}]$:};
    \draw (0,4) node [fill=white] {0};
    \draw (1,4) node [fill=white] {1};
    \draw (2,4) node [fill=white] {1};
    \draw (3,4) node [fill=white] {1};
    \draw (4,4) node [fill=white] {1};
    \draw (5,4) node [fill=white] {1};
    \draw (6,4) node [fill=white] {1};
    \draw (7,4) node [fill=white] {1};
    
    \draw (-5,2) node [fill=white,anchor=west] {$S[\{17,18,\dots,32\}]$:};
    \draw (-0.5,2) node [fill=white] {1};
    \draw (0,2) node [fill=white] {1};
    \draw (0.5,2) node [fill=white] {1};
    \draw (1,2) node [fill=white] {0};
    \draw (1.5,2) node [fill=white] {1};
    \draw (2,2) node [fill=white] {1};
    \draw (2.5,2) node [fill=white] {0};
    \draw (3,2) node [fill=white] {1};
    \draw (3.5,2) node [fill=white] {1};
    \draw (4,2) node [fill=white] {0};
    \draw (4.5,2) node [fill=white] {0};
    \draw (5,2) node [fill=white] {1};
    \draw (5.5,2) node [fill=white] {0};
    \draw (6,2) node [fill=white] {1};
    \draw (6.5,2) node [fill=white] {0};
    \draw (7,2) node [fill=white] {0};
    
    \draw (-5,0) node [fill=white,anchor=west] {$S[\{33,34,\dots,48\}_{\operatorname{even}}]$:};
    \draw (0,0) node [fill=white] {0};
    \draw (1,0) node [fill=white] {1};
    \draw (2,0) node [fill=white] {1};
    \draw (3,0) node [fill=white] {0};
    \draw (4,0) node [fill=white] {0};
    \draw (5,0) node [fill=white] {1};
    \draw (6,0) node [fill=white] {0};
    \draw (7,0) node [fill=white] {1};
    \end{tikzpicture}
    \caption{A string with $48$ characters. For $L=\{2\}$ and $R=\{33,34,\dots,48\}$, for each index of $R_{\operatorname{even}}$, there is only one candidate $(i,d,k)$ for forming an $L$-$R$-$3$-cadence.}
    \label{fig:LRExample}
\end{figure}
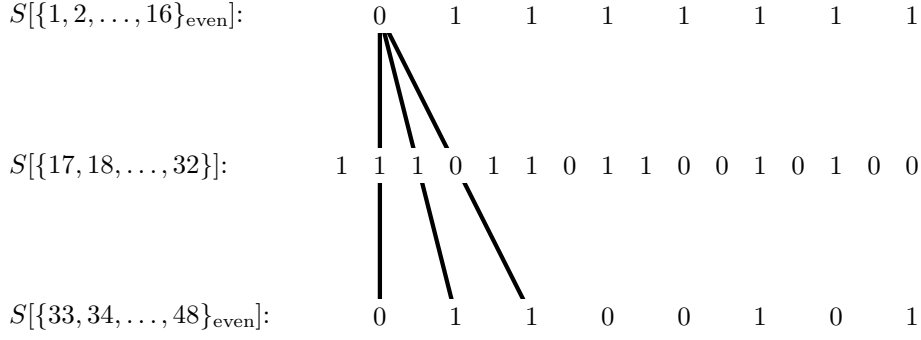

The more difficult case is that one of the two subsequences, without loss of generality $S[R_{\operatorname{even}}]$, is more complex. I.e.\ it consists of multiple runs of $0$s and $1$s and thereby contains both substrings $01$ and $10$. In this case, in order to avoid $L$-$R$-$3$-cadences, the other subsequence $S[L_{\operatorname{even}}]$ is a power of a single character, without loss of generality $0$. This can be checked in linear time in uncompressed strings and in polynomial time in grammar-compressed strings.

It should be no surprise that this case is more difficult since this case occurred in the proof of the $\mathcal{NP}$-completeness of the compressed $L$-$R$-$3$-cadence detection problem. Figure \ref{fig:LRExample} shows that if $L$ is a short interval, we have to check linearly many pairs with respect to $n$ in order to find an $L$-$R$-$3$-cadence. In order to develop a polynomial time algorithm for grammar-compressed strings, we have to use that $L$ is roughly as long as $R$.

By definition of the $L$-$R$-$3$-cadence we get the following lemma:

\begin{lemma} \label{lem:LR0}
    Let $S$ be a binary string and $L$ and $R$ be two intervals. Let further $S[L_{\operatorname{even}}]$ be of the form $0^i$. Define $l_{\min} = \min(L_{\operatorname{even}})$, $l_{\max} = \max(L_{\operatorname{even}})$, $r_{\min} = \min(R_{\operatorname{even}})$ and $r_{\max} = \max(R_{\operatorname{even}})$.
    
    Then, for any $r_0\in R_{\operatorname{even}}$ with $S[r_0]=0$ there is an even $L$-$R$-$3$-cadence which uses this $0$ as last element if and only if
    \[S\left[\left\{\frac{l_{\min}+r_{0}}{2}, \frac{l_{\min}+r_{0}}{2}+1, \dots, \frac{l_{\max}+r_{0}}{2}\right\}\right]\neq 1^{\left(\frac{l_{\max}+r_{0}}{2}-\frac{l_{\min}+r_{0}}{2}+1\right)}\]
    holds.
    
    Conversely, for any $m_0\in \left\{\frac{l_{\min}+r_{\min}}{2}, \frac{l_{\min}+r_{\min}}{2}+1, \dots, \frac{l_{\max}+r_{\max}}{2}\right\}$ with $S[m_0]=0$ there is an even $L$-$R$-$3$-cadence which uses this $0$ as middle element if and only if
    \[S\left[\left\{\max\left(2m_0-l_{\max},r_{\min}\right), \max\left(2m_0-l_{\max},r_{\min}\right)+2, \dots, \min\left(2m_0-l_{\min},r_{\max}\right)\right\}\right]\]
    is not of the form $1^j$.
\end{lemma}

With Corollary \ref{cor:LReasiness}, Lemma \ref{lem:LR01} and Lemma \ref{lem:LR0}, it is possible to efficiently either find an $L$-$R$-$3$-cadence or to shorten the complex interval $R$ without removing any $L$-$R$-$3$-cadences.

\begin{corollary} \label{cor:LRalgo}
    Let $S$ be a binary string and $L$ and $R$ be two intervals. Let further $S[L_{\operatorname{even}}]$ be of the form $0^i$. Define $l_{\min} = \min(L_{\operatorname{even}})$, $l_{\max} = \max(L_{\operatorname{even}})$, $r_{\min} = \min(R_{\operatorname{even}})$ and $r_{\max} = \max(R_{\operatorname{even}})$.
    
    If $S[R_{\operatorname{even}}]$ is of the form $1^j$, there is no even $L$-$R$-$3$-cadence.
    
    Otherwise, define $r_0=\min\left(r\in R_{\operatorname{even}}| S[r]=0\right)$.
    If $S\left[\left\{\frac{l_{\min}+r_{0}}{2}, \frac{l_{\min}+r_{0}}{2}+1, \dots, \frac{l_{\max}+r_{0}}{2}\right\}\right]$ contains a $0$, the corresponding index forms an $L$-$R$-$3$-cadence with an index of $L_{\operatorname{even}}$ and $r_0$.
    
    Otherwise, there is no $L$-$R$-$3$-cadence using $r_0$ as third index, and furthermore, if the substring $S\left[\left\{\frac{l_{\max}+r_{0}}{2}+1, \frac{l_{\max}+r_{0}}{2}+2, \dots, \frac{l_{\max}+r_{\max}}{2}\right\}\right]$ of $S$ is of the form $1^j$, then there is no even $L$-$R$-$3$-cadence.
    
    Otherwise, define $m_0=\min\left(m\in\left\{\frac{l_{\max}+r_{0}}{2}+1, \frac{l_{\max}+r_{0}}{2}+2, \dots, \frac{l_{\max}+r_{\max}}{2}\right\}| S[m]=0\right)$.
    If $S\left[\left\{2m_0-l_{\max}, 2m_0-l_{\max}+1, \dots, \min\left(2m_0-l_{\min},r_{\max}\right)\right\}\right]$ contains a $0$, the corresponding index forms an $L$-$R$-$3$-cadence with an index of $L_{\operatorname{even}}$ and $m_0$.
    
    Otherwise, define $R'=R \cap \mathbb{Z}_{>2m_0-l_{\min}}$. There is an even $L$-$R$-$3$-cadence if and only if there is an even $L$-$R'$-$3$-cadence.
\end{corollary}

An application of this corollary can be seen in Figure \ref{fig:LRExample2}.

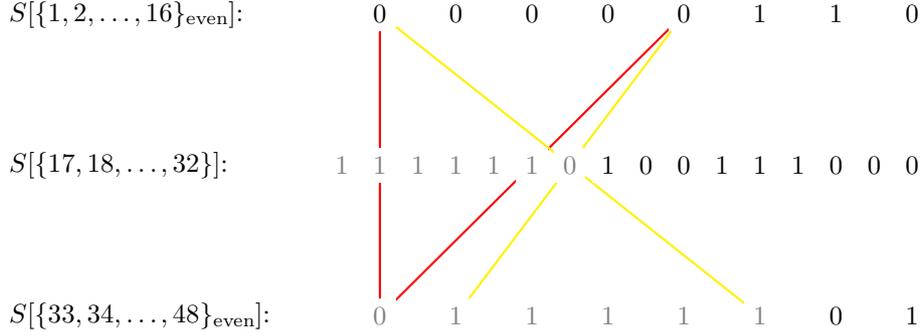
\begin{figure}
    \centering
    \begin{tikzpicture}[scale=1]
    \draw [thick, color=red] (0,4) -- (0,0) -- (4,4);
    
    \draw [thick, color=yellow] (0,4) -- (5,0);
    \draw [thick, color=yellow] (4,4) -- (1,0);
    
    
    \draw (-5,4) node [fill=white,anchor=west] {$S[\{1,2,\dots,16\}_{\operatorname{even}}]$:};
    \draw (0,4) node [fill=white] {0};
    \draw (1,4) node [fill=white] {0};
    \draw (2,4) node [fill=white] {0};
    \draw (3,4) node [fill=white] {0};
    \draw (4,4) node [fill=white] {0};
    \draw (5,4) node [fill=white] {1};
    \draw (6,4) node [fill=white] {1};
    \draw (7,4) node [fill=white] {0};
    
    \draw (-5,2) node [fill=white,anchor=west] {$S[\{17,18,\dots,32\}]$:};
    \draw (-0.5,2) node [fill=white, text=gray] {1};
    \draw (0,2) node [fill=white, text=gray] {1};
    \draw (0.5,2) node [fill=white, text=gray] {1};
    \draw (1,2) node [fill=white, text=gray] {1};
    \draw (1.5,2) node [fill=white, text=gray] {1};
    \draw (2,2) node [fill=white, text=gray] {1};
    \draw (2.5,2) node [fill=white, text=gray] {0};
    \draw (3,2) node [fill=white] {1};
    \draw (3.5,2) node [fill=white] {0};
    \draw (4,2) node [fill=white] {0};
    \draw (4.5,2) node [fill=white] {1};
    \draw (5,2) node [fill=white] {1};
    \draw (5.5,2) node [fill=white] {1};
    \draw (6,2) node [fill=white] {0};
    \draw (6.5,2) node [fill=white] {0};
    \draw (7,2) node [fill=white] {0};
    
    \draw (-5,0) node [fill=white,anchor=west] {$S[\{33,34,\dots,48\}_{\operatorname{even}}]$:};
    \draw (0,0) node [fill=white, text=gray] {0};
    \draw (1,0) node [fill=white, text=gray] {1};
    \draw (2,0) node [fill=white, text=gray] {1};
    \draw (3,0) node [fill=white, text=gray] {1};
    \draw (4,0) node [fill=white, text=gray] {1};
    \draw (5,0) node [fill=white, text=gray] {1};
    \draw (6,0) node [fill=white] {0};
    \draw (7,0) node [fill=white] {1};
    \end{tikzpicture}
    \caption{A string with $48$ characters after one application of Corollary \ref{cor:LRalgo}. Let $L=\{2,3,\dots, 10\}$ and $R=\{33,34,\dots,48\}$ be given. First, the index $r_0=34$ is found. The minimal and maximal candidates for $3$-cadences with $r_0$ are marked with red. Then, the index $m_0=23$ is found. The minimal and maximal candidates for $3$-cadences with $m_0$ are marked with yellow. Afterwards, the gray characters are guaranteed not to form a $3$-cadence with characters from the first run of the string.}
    \label{fig:LRExample2}
\end{figure}

By construction, the set $R'$ contains only elements greater than $2m_0-l_{\min}$. Therefore, either $2m_0-l_{\min}\geq r_{\max}$ and $R'$ is the empty set or $2m_0-l_{\min}< r_{\max}$ and $R'_{\operatorname{even}}$ contains at least
\[\left\lfloor\frac{1+2m_0-l_{\min}-r_{\min}}{2}\right\rfloor\geq \left\lfloor\frac{1+l_{\max}+r_0+2-l_{\min}-r_{\min}}{2}\right\rfloor > \frac{l_{\max}-l_{\min}}{2}\]
elements less than $R_{\operatorname{even}}$. Therefore, the algorithm described in Corollary \ref{cor:LRalgo} has to be used at most $\mathcal{O}\left(\frac{|R|}{|L|}\right)$ times.

On the other hand, the algorithm needs $\mathcal{O}\left(|L|+(r_0-r_{\min})+(m_0-\frac{l_{\min}+r_{\min}}{2})\right)$ time for uncompressed strings and polynomial time for grammar-compressed strings.

Also Corollary \ref{cor:LRalgo} removes at least $\frac{r_0-r_{\min}}{2}$ and at least $m_0-\frac{l_{\min}+r_{\min}}{2}$ from the set $R_{\operatorname{even}}$. Therefore, even if $L$ is small and  either $r_0-r_{\min}$ or $m_0-\frac{l_{\min}+r_{\min}}{2}$ is large, the detection of even $L$-$R$-$3$-cadences can be done in $\mathcal{O}\left(|L|+|R|\right)$ time in uncompressed strings.

By symmetry, the detection of odd $L$-$R$-$3$-cadences can also be done as the detection of even $L$-$R$-$3$-cadences.

This concludes the proof of Theorem \ref{thm:LRCad-P}.

\section{3-Cadences in Binary Strings}

In this section, I will show that the results of Theorem \ref{thm:LRCad-P} also hold for the corresponding $3$-cadence problems:

\begin{theorem}
The decision problem of $3$-cadence detection in binary grammar-compressed strings can be solved in polynomial time.

The decision problem of $3$-cadence detection in binary uncompressed strings can be solved linear time.
\end{theorem}

The main idea of the algorithm of Funakoshi and Pape-Lange in \cite{FPL_3-Cadences} for counting $3$-cadences in uncompressed strings was counting $L$-$R$-$3$-cadences for many pairs of $L$ and $R$. Therefore, we can use the detection algorithm for $L$-$R$-$3$-cadences given by Corollary \ref{cor:LRalgo} in order to detect a $3$-cadence in uncompressed binary strings in $\mathcal{O}(n \log n)$ time. 

However, since this algorithm uses $\Theta(n)$ pairs of $L$ and $R$, this approach does not translate into a polynomial time detection algorithm for $3$-cadences in compressed binary strings. Therefore, instead of dissecting the problem of $3$-cadence detection into many problems of $L$-$R$-$3$-cadence detection, we have to apply the ideas from the last section directly to the problem of $3$-cadence detection. 

Similarly to the $L$-$R$-$3$cadences, there are even $3$-cadences and odd $3$-cadences. Without loss of generality, this paper only considers the even $3$-cadences and defines:

\begin{definition}
    An even $3$-cadence is a $3$-cadence which starts with an even index.
    
    We define the string $S_{\operatorname{even}}$ by $S_{\operatorname{even}} = S\left[\left\{2,4,6,\dots,2\left\lfloor\frac{|S|}{2}\right\rfloor\right\}\right]$ to be the restriction of $S$ to the characters with even indices.
\end{definition}

Let $i,d$ be two integers such that $i-d\leq 0$ and $i+3d>n$ hold. Let $L=\{1,2,\dots, i\}$ and $R=\{i+2d, i+2d+1,\dots, n\}$ be two intervals. Then each $L$-$R$-$3$-cadence is also a $3$-cadence. On the other hand, each $3$-cadence defines integers $i$ and $d$ such that $i-d\leq 0$ and $i+3d>n$ hold. Therefore, we can use Lemma \ref{lem:LR2} to obtain that if $S$ has a $3$-cadence, it also has a $3$-cadence that either starts in one of the first two runs of $S_{\operatorname{even}}$ or ends in one of the last two runs of $S_{\operatorname{even}}$.

The main challenge for the adaption of the detection algorithm for $L$-$R$-$3$-cadences to an detection algorithm for $3$-cadences is that Lemma \ref{lem:LR0} does not quite work.

See, for example, the string $S=000100011$. Since a $3$-cadence can start anywhere in the first third of the string and can end anywhere in the last third of the string, we have $L=\{1,2,3\}$ and $R=\{7,8,9\}$.

In terms of $L$-$R$-$3$-cadences, if we ignore the actual characters of $S$, the $0$ at index $7$ can form an $L$-$R$-$3$-cadence with the index $1$ as well as with the index $3$. Of these two possibilities, only the arithmetic progression $3,5,7$ forms an $L$-$R$-$3$-cadence. However, this arithmetic progression is not structurally maximal and hence not a $3$-cadence.

On the other hand, in the string $S'=000100110$, the $0$ at index $9$ can form an $L$-$R$-$3$-cadence as well as a $3$-cadence with the index $1$ as well as with the index $3$. Therefore, the arithmetic progression $3,6,9$ forms an $L$-$R$-$3$-cadence as well as a $3$-cadence.

We therefore have to restrict the strings in Lemma \ref{lem:LR0} to those indices such that the corresponding $3$-sub-cadences are structurally maximal. We assume without loss of generality that $S[2]=0$ holds.

\begin{lemma} \label{lem:30}
    Let $S$ be a binary string. Define the two intervals $L=\left\{1,2,\dots,\left\lfloor\frac{1}{3}n\right\rfloor\right\}$ and $R=\left\{\left\lfloor\frac{2}{3}n\right\rfloor+1,\left\lfloor\frac{2}{3}n\right\rfloor+2,\dots,n \right\}$. 
    
    Let further $S[L_{\operatorname{even}}]$ be of the form $0^i S'$ and let $l_{\min} = \min(L_{\operatorname{even}}) = 2$ be the first index of $L_{\operatorname{even}}$ and 
    $l_{\max} = 2+2(i-1) = 2i$ be the index corresponding to the last $0$ in the first run of $S[L_{\operatorname{even}}]$.
    
    Then, for any $r_0\in R_{\operatorname{even}}$ with $S[r_0]=0$ define $l'_{\max} = \min \left(l_{\max}, 2\left\lfloor \frac{r_0}{6}\right\rfloor, 2(\left\lceil \frac{3r_0-2n}{2}\right\rceil -1) \right)$. There is an even $3$-cadence which uses this $0$ as last element and any of the first $i$ $0$s of $S[L_{\operatorname{even}}]$ as first element if and only if
    \[S\left[\left\{\frac{l_{\min}+r_{0}}{2}, \frac{l_{\min}+r_{0}}{2}+1, \dots, \frac{l'_{\max}+r_{0}}{2}\right\}\right]\neq 1^{\left(\frac{l'_{\max}+r_{0}}{2}-\frac{l_{\min}+r_{0}}{2}+1\right)}\]
    holds.
    
    Conversely, for any $m_0\in \left\{\frac{l_{\min}+r_{\min}}{2}, \frac{l_{\min}+r_{\min}}{2}+1, \dots, \frac{l_{\max}+r_{\max}}{2}\right\}$ with $S[m_0]=0$ define 
    $l''_{\min} = \max \left(l_{\min}, 2\left(m_0 - \left\lfloor \frac{n}{2}\right\rfloor\right)\right)$ and $l''_{\max} = \min \left(l_{\max}, 2\left\lfloor \frac{m_0}{4}\right\rfloor, 2\left(\left\lceil \frac{3m_0-n}{4}\right\rceil -1\right) \right)$. There is an even $3$-cadence which uses this $0$ as middle element if and only if
    \[S\left[\left\{2m_0-l''_{\max}, 2m_0-l''_{\max}+2, \dots, 2m_0-l''_{\min}\right\}\right] \neq 1^{\left((2m_0-l''_{\min})-(2m_0-l''_{\max})+1\right)}\]
    holds.
\end{lemma}
\begin{proof}
    Like Lemma \ref{lem:LR0}, this lemma basically holds by definition of the $3$-cadence.
    
    All indices less than or equal to $l'_{\max}$ can form a $3$-cadence with $r_0$, since
    $l_1 := 2\left\lfloor \frac{r_0}{6}\right\rfloor$ is the largest even index fulfilling the inequality $l_{1} - \frac{r_0-l_{1}}{2} \leq 0$ and $l_2 := 2(\left\lceil \frac{3r_0-2n}{2}\right\rceil -1)$ is the largest even index fulfilling the inequality $l_{2} + 3\frac{r_0-l_{2}}{2} > n$.
    
    Similarly, all indices between $l''_{\min}$ and $l''_{\max}$ can form a $3$-cadence with $m_0$, since the index
    $l_3 := 2\left(m_0 - \left\lfloor \frac{n}{2}\right\rfloor\right)$ is the smallest even index fulfilling the inequality $l_{3} + 2(m_0-l_{3}) \leq n$,
    the index $l_4 := 2\left\lfloor \frac{m_0}{4}\right\rfloor$ is the largest even index fulfilling the inequality $l_{4} - (m_0-l_{4}) \leq 0$ and
    $l_5 := 2\left(\left\lceil \frac{3m_0-n}{4}\right\rceil -1\right)$ is the largest even index fulfilling the inequality $l_{5} + 3(m_0-l_{5}) > n$.
\end{proof}

Similarly to the case of the $L$-$R$-$3$-cadence, we can use this lemma to shrink the interval in which the last element of the arithmetic progression can be.

\begin{corollary} \label{cor:3algo}
    Let $S$ be a binary string. Define the two intervals $L=\left\{1,2,\dots,\left\lfloor\frac{1}{3}n\right\rfloor\right\}$ and $R=\left\{r_{\min},r_{\min}+1,\dots,n \right\}$ for a $r_{\min}\geq \left\lfloor\frac{2}{3}n\right\rfloor+1$. 
    
    Let further $S[L_{\operatorname{even}}]$ be of the form $0^i S'$ and $L'_{\operatorname{even}}$ be the set of indices of the first run in $S[L_{\operatorname{even}}]$.
    Let $l_{\min} = \min(L'_{\operatorname{even}}) = 2$ be the first index of $L_{\operatorname{even}}$
    and $l_{\max} = \max(L'_{\operatorname{even}}) = 2+2(i-1) = 2i$ be the index corresponding to the last $0$ in the first run of $S[L_{\operatorname{even}}]$.
    
    If $S[R_{\operatorname{even}}]$ is of the form $1^j$, there is no even $3$-cadence using an element of $L'_{\operatorname{even}}$ as first index.
    
    Otherwise, define $r_0=\min\left(r\in R_{\operatorname{even}}| S[r]=0\right)$ and the corresponding maximal index for the first element of the $3$-cadence $l'_{\max} = \min \left(l_{\max}, 2\left\lfloor \frac{r_0}{6}\right\rfloor, 2(\left\lceil \frac{3r_0-2n}{2}\right\rceil -1) \right)$.
    
    If $S\left[\left\{\frac{l_{\min}+r_{0}}{2}, \frac{l_{\min}+r_{0}}{2}+1, \dots, \frac{l'_{\max}+r_{0}}{2}\right\}\right]$ contains a $0$, the corresponding index forms a $3$-cadence with an index of $L'_{\operatorname{even}}$ and $r_0$.
    
    Otherwise, there is no $3$-cadence using $r_0$ as third index and a $0$ from the first run of $S[L_{\operatorname{even}}]$ as first index and hence, if the substring $S\left[\left\{\frac{l'_{\max}+r_{0}}{2}+1, \frac{l'_{\max}+r_{0}}{2}+2, \dots, \frac{l_{\max}+r_{\max}}{2}\right\}\right]$ of $S$ is of the form $1^j$, then there is no even $3$-cadence using an element of $L'_{\operatorname{even}}$ as first index.
    
    Otherwise, define $m_0=\min\left(m\in\left\{\frac{l'_{\max}+r_{0}}{2}+1, \frac{l'_{\max}+r_{0}}{2}+2, \dots, \frac{l_{\max}+r_{\max}}{2}\right\}| S[m]=0\right)$, the index $l''_{\min} = \max \left(l_{\min}, 2\left(m_0 - \left\lfloor \frac{n}{2}\right\rfloor\right)\right)$ 
    and 
    $l''_{\max} = \min \left(l_{\max}, 2\left\lfloor \frac{m_0}{4}\right\rfloor, 2\left(\left\lceil \frac{3m_0-n}{4}\right\rceil -1\right) \right)$.
    
    If $S\left[\left\{2m_0-l''_{\max}, 2m_0-l''_{\max}+2, \dots, 2m_0-l''_{\min}\right\}\right]$ contains a $0$, the corresponding index forms an $3$-cadence with an index of $L'_{\operatorname{even}}$ and $m_0$.
    
    Otherwise, define $R'=R \cap \mathbb{Z}_{>2m_0-l''_{\min}}$. There is an even $3$-cadence using an element of $L'_{\operatorname{even}}$ as first index if and only if there is an even $3$-cadence using an element of $L'_{\operatorname{even}}$ as first index and an element of $R'_{\operatorname{even}}$ as last index.
\end{corollary}

An application of this corollary can be seen in Figure \ref{fig:3Example}.

\begin{figure}
    \centering
    \begin{tikzpicture}[scale=1]
    \draw [thick, color=red] (0,4) -- (0,0) -- (1,4);
    
    \draw [thick, color=yellow] (0,4) -- (2,0);
    \draw [thick, color=yellow] (1,4) -- (1,0);
    
    
    \draw (-5,4) node [fill=white,anchor=west] {$S[\{1,2,\dots,16\}_{\operatorname{even}}]$:};
    \draw (0,4) node [fill=white] {0};
    \draw (1,4) node [fill=white] {0};
    \draw (2,4) node [fill=white] {0};
    \draw (3,4) node [fill=white] {0};
    \draw (4,4) node [fill=white] {0};
    \draw (5,4) node [fill=white] {1};
    \draw (6,4) node [fill=white] {1};
    \draw (7,4) node [fill=white] {0};
    
    \draw (-5,2) node [fill=white,anchor=west] {$S[\{17,18,\dots,32\}]$:};
    \draw (-0.5,2) node [fill=white, text=gray] {1};
    \draw (0,2) node [fill=white, text=gray] {1};
    \draw (0.5,2) node [fill=white, text=gray] {1};
    \draw (1,2) node [fill=white, text=gray] {0};
    \draw (1.5,2) node [fill=white] {1};
    \draw (2,2) node [fill=white] {1};
    \draw (2.5,2) node [fill=white] {0};
    \draw (3,2) node [fill=white] {1};
    \draw (3.5,2) node [fill=white] {0};
    \draw (4,2) node [fill=white] {0};
    \draw (4.5,2) node [fill=white] {1};
    \draw (5,2) node [fill=white] {1};
    \draw (5.5,2) node [fill=white] {1};
    \draw (6,2) node [fill=white] {0};
    \draw (6.5,2) node [fill=white] {0};
    \draw (7,2) node [fill=white] {0};
    
    \draw (-5,0) node [fill=white,anchor=west] {$S[\{33,34,\dots,48\}_{\operatorname{even}}]$:};
    \draw (0,0) node [fill=white, text=gray] {0};
    \draw (1,0) node [fill=white, text=gray] {1};
    \draw (2,0) node [fill=white, text=gray] {1};
    \draw (3,0) node [fill=white] {0};
    \draw (4,0) node [fill=white] {1};
    \draw (5,0) node [fill=white] {1};
    \draw (6,0) node [fill=white] {0};
    \draw (7,0) node [fill=white] {1};
    \end{tikzpicture}
    \caption{A string with $48$ characters after one application of Corollary \ref{cor:3algo}. First, the index $r_0=34$ is found. The minimal and maximal candidates for $3$-cadences with $r_0$ are marked with red. Then, the index $m_0=20$ is found. The minimal and maximal candidates for $3$-cadences with $m_0$ are marked with yellow. Afterwards, the gray characters are guaranteed not to form a $3$-cadence with characters from the first run of $S_{\operatorname{even}}$.}
    \label{fig:3Example}
\end{figure}
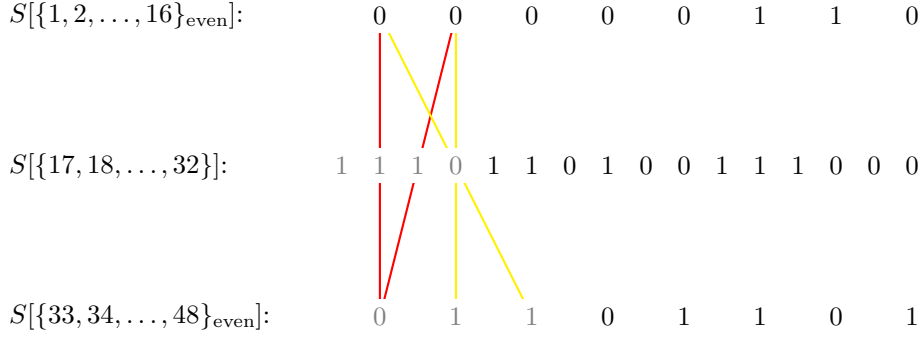

In the uncompressed case, each element of the middle third and the last third has to be read at most once in order to decide whether there is a $3$-cadence which starts in the first run of $S_{\operatorname{even}}$. Furthermore, we can modify this algorithm to detect the existence of a $3$-cadence which start in the second run of $S_{\operatorname{even}}$. By symmetry, we can also decide in linear time, whether there exists a $3$-cadence which end in one of the two last runs of $S_{\operatorname{even}}$. Similarly, we can decide in linear time, whether there is an odd $3$-cadence.

This implies:

\begin{theorem}
    Let $S$ be a binary string. We can decide in linear time whether $S$ contains a $3$-cadence.
    
    If there is a $3$-cadence, this algorithm can also return such a cadence.
\end{theorem}

In the compressed case, if the first run of $S_{\operatorname{even}}$ contains only a single index, the corresponding $3$-cadences are exactly the $L$-$R$-$3$-cadences with $L=\{1\}$ and $R=\left\{\left\lfloor\frac{2}{3}n\right\rfloor+1,\left\lfloor\frac{2}{3}n\right\rfloor+2,\dots,n \right\}$. Since the detection of these $L$-$R$-$3$-cadences is $\mathcal{NP}$-complete, if $\mathcal{P}\neq\mathcal{NP}$ holds, it is not possible to decide in polynomial time, whether there is a $3$-cadence which starts in the first run.

Luckily, it is not necessary to decide whether there is a $3$-cadence which starts in the first run in order to decide whether there is a $3$-cadence at all. Let $S[L_{\operatorname{even}}]$ of the form $0^i 1 S'$. Then the $1$ has index $2i+2$. Let $r_1$ be the smallest even index such that $(2i+2)-\frac{r_1-(2i+2)}{2}\leq 0$ and $(2i+2)+3\frac{r_1-(2i+2)}{2}> n$ hold. I.e.\ the smallest index such that the arithmetic progression $2i+2, \frac{2i+2+r_1}{2}, r_1$ could form a $3$-cadence if the three characters in the underlying string were equal.

In this setting, for $L=\{1,2,\dots, 2i+2\}$ and $R=\{r_1,r_1 +1,\dots, n\}$, each $L$-$R$-$3$-cadence is structurally maximal and therefore forms a $3$-cadence. This implies that we can check in polynomial time whether $S[R_{\operatorname{even}}]$ contains the substring 10 and therefore, whether such a $3$-cadence exists. If such a $3$-cadence exists, we are done. 

Otherwise, $S[R_{\operatorname{even}}]$ is of the form $0^j 1^{j'}$ with $j,j'\geq 0$ and we can use Lemma \ref{lem:LR01} to check in polynomial time, whether such a $3$-cadence exists. We can then use the symmetry to detect a $3$-cadence which ends at one of the last $j'+1$ characters of $S[R_{\operatorname{even}}]$.

It is left to show that even in the compressed case, Corollary \ref{cor:3algo} is fast enough to allow finding the $3$-cadences which start in the first run of $S_{\operatorname{even}}$ and end at an index smaller than $r_1$.

In order to do this, I will show that with each application of Corollary \ref{cor:3algo} the discarded part of $R$ doubles until $r_1$ is reached. In the worst case, the new $0$ at $r_0$ is directly at the beginning of $S[R_{\operatorname{even}}]$. Since each $3$-sub-cadences with distance greater than or equal to $\frac{1}{3}n$ is a $3$-cadence, we can assume $r_0<r_1\leq 2i+2+\frac{2}{3}n$ holds and therefore $l_{\max} \geq r_0-\frac{2}{3}n$ holds as well. Also, both $2\left\lfloor \frac{r_0}{6}\right\rfloor$ and $2(\left\lceil \frac{3r_0-2n}{2}\right\rceil -1)$ are greater than or equal to $r_0-\frac{2}{3}n$. Therefore $l'_{\max} \geq r_0-\frac{2}{3}n$ holds. 

Similarly, in the worst case, the new $0$ at $m_0$ is directly behind $\frac{l'_{\max}+r_0}{2}\geq r_0 - \frac{1}{3}n$. Therefore, in the worst case, the index $m_0$ is at $2r_0 - \frac{1}{3}n + 1$. With $l''_{\min} = \max \left(l_{\min}, 2\left(m_0 - \left\lfloor \frac{n}{2}\right\rfloor\right)\right)$, this implies that the inequality $2m_0-l''_{\min}\geq \min\left(r_0 + (r_0-\frac{2}{3}n), n-1\right)$ holds.

Hence, under the assumption that $r_0<r_1$ holds, one application of Corollary \ref{cor:3algo} checks for an interval of size $r_0-\frac{2}{3}n$, whether this interval contains any last indices for a $3$-cadence which starts in the first run. Therefore, we only need at most $\log n$ applications of this corollary.

This implies that it can be decided in polynomial time whether a grammar-compressed binary string contains any $3$-cadences.

\section{Conclusion}

This paper shows that we can decide in linear time whether an uncompressed binary string contains a $3$-cadence. While we should expect that it is more difficult to avoid $3$-cadences in binary strings than to include $3$-cadences, it is surprising that it is strictly easier to decide whether there is any $3$-cadence at all than to decide whether there is a $3$-cadence with a given character.

For the latter problem, Amir et al.\ have shown in \cite{Cadences_Amir} that we should not expect a solution with time complexity $o(n\log n)$ by reduction of the $3$SUM problem with bounded elements.

For the compressed case, we have shown that we can decide in polynomial time whether a grammar-compressed binary string contains a $3$-cadence. However, all even slightly harder problems have been shown to be $\mathcal{NP}$-complete.
These hardness-results seem to indicate that cadences may not be very useful in compressed pattern matching.

While we can decide in constant time whether a string contains a $k$-sub-cadence, there are no known nontrivial bounds on the bit complexity of the detection of $k$-sub-cadences with a given character. Closely related, it is unknown whether equidistant subsequence matching is $\mathcal{NP}$-hard on compressed binary strings.

Finally, in terms of uncompressed cadence detection, it is still unknown whether we can decide with sub-quadratic bit complexity whether a given string contains a $4$-cadence. The currently best result is by Funakoshi et al., who presented a detection algorithm with sub-quadratic time complexity in the word RAM model in \cite{Funakoshi_ESM}.

\bibliography{cadences}

\appendix

\section{L-R-Cadences with Overlap}

In this section, I will extend the result of Section \ref{sec:LR} and show that all results still hold if the intervals $L$ and $R$ are allowed to have overlap.

Since all $\mathcal{NP}$-complete of Section \ref{sec:LR} are still $\mathcal{NP}$-complete for this more general notion of $L$-$R$-cadences, it is left to show that one can detect an $L$-$R$-3-cadence in linear time in uncompressed binary strings with respect to the length of the string and in polynomial time in grammar-compressed binary strings with respect to the compressed size of the string and the additional variable $\max\left(\frac{|L|}{|R|},\frac{|R|}{|L|}\right)$.

Let $L$ and $R$ therefore be two overlapping intervals. Define the overlapping part $M = L\cap R$ as well as the two non-overlapping parts $L' = L\setminus M$ and $R' = R\setminus M$. Since each $L$-$R$-$3$-cadence starts with an index in $L$ and ends with an index in $R$, we can assume that each index in $L'$ is less than each index in $M$ and that each index in $M$ is less than each index in $R'$.

By construction, each $L$-$R$-$3$-cadence is either 
\begin{itemize}
    \item an $L'$-$M$-$3$-cadence,
    \item an $L'$-$R'$-$3$-cadence, 
    \item an $M$-$M$-$3$-cadence or 
    \item an $M$-$R'$-$3$-cadence.
\end{itemize}

The $M$-$M$-$3$-cadences are exactly the $3$-sub-cadences on the string $S[M]$. Therefore, van der Waerden's theorem shows that if $S[M]$ contains at least $9$ characters, it is guaranteed that a $M$-$M$-$3$-cadences exist and that we can find such a sub-cadence by reading the first $9$ characters of $S[M]$.

We can therefore assume in the remainder of the proof that $S[M]$ contains less than $9$ characters. In this case, we can find all $M$-$M$-$3$-cadences in constant time in uncompressed strings and in linear time in grammar-compressed strings.

Since $|M|$ is small, there is no detection algorithm for $L'$-$M$-$3$-cadences and for $M$-$R'$-$3$-cadences which runs in polynomial time in grammar-compressed strings. However, each $L'$-$R$-$3$-cadence is either an $L'$-$M$-$3$-cadence or an $L'$-$R'$-$3$-cadence.

If $L'$ is empty, then there are no $L'$-$R$-$3$-cadences. Otherwise, we can use the results of Section \ref{sec:LR} to detect an $L'$-$R$-$3$-cadence. This takes linear time in uncompressed strings. Furthermore, in grammar-compressed strings, we can detect an $L'$-$R$-$3$-cadence in polynomial time with respect to the compressed size and $\max\left(\frac{|L'|}{|R|},\frac{|R|}{|L'|}\right)$. However, since $|L|-|M|=|L'|\neq 0$ and $|M|<9$ hold, the fraction $\frac{|R|}{|L'|}$ is bounded from above by $9\frac{|R|}{|L|}$. Therefore, in grammar-compressed strings, we can detect an $L'$-$R$-$3$-cadence in polynomial time with respect to the compressed size and $\max\left(\frac{|L|}{|R|},\frac{|R|}{|L|}\right)$.

Since, each $L$-$R'$-$3$-cadence is either an $L'$-$R'$-$3$-cadence or an $M$-$R'$-$3$-cadence, we can similarly detect those sub-cadences. Furthermore, since we do not attempt to count the number of $L$-$R$-$3$-cadences, it is not a problem that we may find $L'$-$R'$-$3$-cadences twice.

This implies:

\begin{theorem}
    For two, not necessarily disjoint, intervals $L$ and $R$, it is possible to detect whether a binary string contains any $L$-$R$-cadence
    \begin{itemize}
        \item in $\mathcal{O}\left(|L|+|R|\right)$ in uncompressed strings and
        \item in polynomial time with respect to the compressed size and $\max\left(\frac{|L|}{|R|},\frac{|R|}{|L|}\right)$ in grammar-compressed strings.
    \end{itemize}
    
    If such an $L$-$R$-cadence exist, we can find such a cadence in the same time.
\end{theorem}

\end{document}